\documentclass{stacs_proc}

\usepackage{mathrsfs,xspace,clrscode} 

\theoremstyle{plain} \newtheorem{fact}[thm]{Fact}
\theoremstyle{definition} \newtheorem{defn}[thm]{Definition}

\newcommand{\scrA}{\mathscr{A}}
\newcommand{\cA}{\mathcal{A}}
\newcommand{\scrB}{\mathscr{B}}
\newcommand{\cB}{\mathcal{B}}
\renewcommand{\b}{\{0,1\}}

\newcommand{\eps}{\varepsilon}

\newcommand{\ind}[1]{\mathbf{1}_{#1}}

\newcommand{\nn}{[n]^{[n]}}

\DeclareMathOperator{\Range}{Range}
\newcommand{\cS}{\mathcal{S}}
\DeclareMathSymbol{\R}{\mathbin}{AMSb}{"52}
\DeclareMathSymbol{\Z}{\mathbin}{AMSb}{"5A}

\newcommand{\accz}{\mathsf{ACC}^0}

\newcommand{\ceq}{\subseteq}
\DeclareMathOperator{\cost}{cost}

\newcommand{\idx}{\textsc{index}\xspace}
\newcommand{\mpj}{\textsc{mpj}\xspace}
\newcommand{\mpjh}{\widehat{\textsc{mpj}}\xspace}
\newcommand{\mpjperm}{{\textsc{mpj}}^{\rm perm}}
\newcommand{\mpjhperm}{\widehat{\textsc{mpj}}^{\rm perm}}
\newcommand{\plr}{\textsc{plr}\xspace}

\newcommand{\tpj}{\textsc{tpj}\xspace}

\newcommand{\ang}[1]{\langle #1 \rangle}
\newcommand{\ceil}[1]{\lceil #1 \rceil}

\newcommand{\comments}[1]{}

\begin{document}

\title[Sublinear Protocols for Pointer Jumping]{%
  Sublinear Communication Protocols for Multi-Party Pointer Jumping
  and a Related Lower Bound
}

\author[aaa]{J.~Brody}{Joshua Brody}
\address[aaa]{%
  Department of Computer Science \newline
  Dartmouth College \newline
  Hanover, NH 03755, USA}
\author[aaa]{A.~Chakrabarti}{Amit Chakrabarti}

\thanks{%
  Work supported in part by an NSF CAREER Award CCF-0448277, NSF
  grants CCF-0514870 and EIA-98-02068. Work partly done while 
  the authors were visiting the University of Washington, Seattle, WA.
}




\begin{abstract}

We study the one-way number-on-the-forehead (NOF) communication
complexity of the $k$-layer pointer jumping problem with $n$ vertices
per layer. This classic problem, which has connections to many aspects
of complexity theory, has seen a recent burst of research activity,
seemingly preparing the ground for an $\Omega(n)$ lower bound, for
constant $k$.  Our first result is a surprising sublinear --- i.e.,
$o(n)$ --- upper bound for the problem that holds for $k \ge 3$, dashing
hopes for such a lower bound.

A closer look at the protocol achieving the upper bound shows that all
but one of the players involved are {\em collapsing}, i.e., their
messages depend only on the composition of the layers ahead of them. We
consider protocols for the pointer jumping problem where {\em all}
players are collapsing. Our second result shows that a strong $n -
O(\log n)$ lower bound does hold in this case.  Our third result is
another upper bound showing that nontrivial protocols for (a non-Boolean
version of) pointer jumping are possible even when all players are
collapsing.

Our lower bound result uses a novel proof technique, different from
those of earlier lower bounds that had an information-theoretic flavor.
We hope this is useful in further study of the problem.

\end{abstract}

\maketitle

\stacsheading{2008}{145-156}{Bordeaux}
\firstpageno{145}



\section{Introduction}

Multi-party communication complexity in general, and the {\em pointer
jumping} problem (also known as the {\em pointer chasing} problem) in
particular, has been the subject of plenty of recent research. This is
because the model, and sometimes the specific problem, bears on several
aspects of computational complexity: among them, circuit
complexity~\cite{Yao90,HastadG91,BeigelT94}, proof size lower
bounds~\cite{BeamePS05} and space lower bounds for streaming
algorithms~\cite{AlonMS99,GuhaM07,ChakrabartiJP08}. The most impressive
known consequence of a strong multi-party communication lower bound
would be to exhibit non-membership in the complexity class $\accz$;
details can be found in Beigel and Tarui~\cite{BeigelT94} or in the
textbook by Arora and Barak~\cite{AroraBarak-book}.  Vexingly, it is not
even known whether or not $\accz = \mathsf{NEXP}$.

The setting of multi-party communication is as follows. There are $k$
players (for some $k \ge 2$), whom we shall call $\plr_1, \plr_2,
\ldots, \plr_k$, who share an input $k$-tuple $(x_1, x_2, \ldots, x_k)$.
The goal of the players is to compute some function $f(x_1, x_2, \ldots,
x_k)$.  There are two well-studied sharing models: the {\em
number-in-hand} model, where $\plr_i$ sees $x_i$, and the {\em
number-on-the-forehead} (NOF) model, where $\plr_i$ sees all $x_j$s such
that $j \ne i$.  Our focus in this paper will be on the latter model,
which was first introduced by Chandra, Furst and
Lipton~\cite{ChandraFL83}. It is in this model that communication lower
bounds imply lower bounds against $\accz$. We shall use $C(f)$ to denote
the deterministic communication complexity of $f$ in this model. Also of
interest are randomized protocols that only compute $f(x)$ correctly
with high probability: we let $R_\eps(f)$ denote the $\eps$-error
randomized communication complexity of $f$. Our work here will stick to
deterministic protocols, which is a strength for our upper bounds.
Moreover, it is not a serious weakness for our lower bound, because the
$\accz$ connection only calls for a deterministic lower bound.

Notice that the NOF model has a feature not seen elsewhere in
communication complexity: the players {\em share} plenty of information.
In fact, for large $k$, each individual player already has ``almost''
all of the input. This intuitively makes lower bounds especially hard to
prove and indeed, to this day, no nontrivial lower bound is known in the
NOF model for any explicit function with $k = \omega(\log n)$ players,
where $n$ is the total input size. The pointer jumping problem is widely
considered to be a good candidate for such a lower bound. As noted by
Damm, Jukna and Sgall~\cite{DammJS98}, it has many natural special
cases, such as shifting, addressing, multiplication and convolution.
This motivates our study.


\subsection{The Pointer Jumping Problem and Previous Results}
\label{sec:prev}

There are a number of variants of the pointer jumping problem. Here we
study two variants: a Boolean problem, $\mpj^n_k$, and a non-Boolean
problem, $\mpjh^n_k$ (henceforth, we shall drop the superscript $n$).
In both variants, the input is a subgraph of a fixed layered graph that
has $k+1$ layers of vertices, with layer $0$ consisting of a single
vertex, $v_0$, and layers $1$ through $k-1$ consisting of $n$ vertices
each (we assume $k \ge 2$). Layer $k$ consists of $2$ vertices in the
case of $\mpj_k$ and $n$ vertices in the case of $\mpjh_k$. The input
graph is a subgraph of the fixed layered graph in which every vertex
(except those in layer $k$) has outdegree $1$. The desired output is the
name of the unique vertex in layer $k$ reachable from $v_0$, i.e., the
final result of ``following the pointers'' starting at $v_0$.  The
output is therefore a single bit in the case of $\mpj_k$ or a
$\ceil{\log n}$-bit string in the case of $\mpjh_k$.\footnote{Throughout
this paper we use ``$\log$'' to denote logarithm to the base $2$.}

The functions $\mpj_k$ and $\mpjh_k$ are made into NOF communication
problems as follows: for each $i \in [k]$, a description of the $i$th
layer of edges (i.e., the edges pointing into the $i$th layer of
vertices) is written on $\plr_i$'s forehead. In other words, $\plr_i$
sees every layer of edges except the $i$th. The players are allowed to
write one message each on a public {\em blackboard} and must do so in
the fixed order $\plr_1, \plr_2, \ldots, \plr_k$. The final player's
message must be the desired output.  Notice that the specific order of
speaking --- $\plr_1, \plr_2, \ldots, \plr_k$ --- is important to make
the problem nontrivial. Any other order of speaking allows an easy
deterministic protocol with only $O(\log n)$ communication.

Consider the case $k = 2$. The problem $\mpj_2$ is equivalent to the
two-party communication problem $\idx$, where Alice holds a bit-vector
$x\in\b^n$, Bob holds an index $i\in[n]$, and Alice must send Bob a
message that enables him to output $x_i$. It is easy to show that
$C(\mpj_2) = n$. In fact, Ablayev~\cite{Ablayev96} shows the tight
tradeoff $R_\eps(\mpj_2) = (1 - H(\eps))n$, where $H$ is the binary
entropy function.  It is tempting to conjecture that this lower bound
generalizes as follows.

\begin{conj} \label{conj:mpj-lin-lb}
  There is a nondecreasing function $\xi:\Z^+\to\R^+$ such that,
  $\forall\,k:~ C(\mpj_k) = \Omega(n/\xi(k))$.
\end{conj}

Note that, by the results of Beigel and Tarui~\cite{BeigelT94}, in order
to show that $\mpj_k \notin \accz$ it would suffice, for instance, to
prove the following (possibly weaker) conjecture.

\begin{conj} \label{conj:mpj-accz}
  There exist constants $\alpha, \beta > 0$ such that, for $k
  = n^\alpha$, $C(\mpj_k) = \Omega(n^\beta)$.
\end{conj}

Conjecture~\ref{conj:mpj-lin-lb} is consistent with (and to an extent
motivated by) research prior to this work. In weaker models of
information sharing than the NOF model, an equivalent statement is known
to be true, even for randomized protocols. For instance, Damm, Jukna and
Sgall~\cite{DammJS98} show an $\Omega(n/k^2)$ communication lower bound
in the so-called {\em conservative} model, where $\plr_i$ has only a
limited view of the layers of the graph behind her: she only sees the
result of following the first $i-1$ pointers.
Chakrabarti~\cite{Chakrabarti07} extends this bound to randomized
protocols and also shows an $\Omega(n/k)$ lower bound in the so-called
{\em myopic} model, where $\plr_i$ has only a limited view of the layers
ahead of her: she cannot see layers $i+2, \ldots, k$. 

For the full NOF model, Wigderson, building on the work of Nisan and
Wigderson~\cite{NisanW93}, showed that $C(\mpj_3) = \Omega(\sqrt n)$.
This result is unpublished, but an exposition can be found in Babai,
Hayes and Kimmel~\cite{BabaiHK01}. Very recently, Viola and
Wigderson~\cite{ViolaW07} generalized this result and extended it to
randomized protocols, showing that $R_{1/3}(\mpj_k) =
\Omega(n^{1/(k-1)}/k^{O(k)})$. Of course, this bound falls far short of
that in Conjecture~\ref{conj:mpj-lin-lb} and does nothing for
Conjecture~\ref{conj:mpj-accz}. However, it is worth noting that the
Viola-Wigderson bound in fact applies to the much smaller subproblem of
{\em tree pointer jumping} (denoted $\tpj_k$), where the underlying
layered graph is a height-$k$ tree, with every vertex in layers $0$
through $k-2$ having $n^{1/(k-1)}$ children and every vertex in layer
$k-1$ having two children. It is easy to see that $C(\tpj_k) =
O(n^{1/(k-1)})$. Thus, one might hope that the more general problem
$\mpj_k$ has a much stronger lower bound, as in
Conjecture~\ref{conj:mpj-lin-lb}.

On the upper bound side, Damm et al.~\cite{DammJS98} 
show 
that
$C(\mpjh_k) = O(n\log^{(k-1)}n)$, where $\log^{(i)}n$ is the $i$th
iterated logarithm of $n$. This improves on the trivial upper bound of
$O(n\log n)$. Their technique does not yield anything nontrivial for the
Boolean problem $\mpj_k$, though. However, Pudlak, R\"odl and
Sgall~\cite{PudlakRS97} obtain a sublinear upper bound of $O(n\log\log
n/\log n)$ for a special case of $\mpj_3$. Their protocol works only
when every vertex in layer $2$ has {\em indegree} $1$, or equivalently,
when the middle layer of edges in the input describes a {\em
permutation} of $[n]$.

\subsection{Our Results}

The protocol of Pudlak et al.~\cite{PudlakRS97} did not rule out
Conjecture~\ref{conj:mpj-lin-lb}, but it did suggest caution. Our first
result is the following upper bound --- in fact the first nontrivial
upper bound on $C(\mpj_k)$ --- that falsifies the conjecture.

\begin{thm} \label{thm:mpj-ub}
  For $k \ge 3$, we have
  \[
    C(\mpj_k) ~=~ O\left( n \left(\frac{k\log\log n}{\log n}\right)
      ^{(k-2)/(k-1)} \right) \, .
  \]
  In particular, $C(\mpj_3) = O(n\sqrt{\log\log n/\log n})$.
\end{thm}

A closer look at the protocol that achieves the upper bound above
reveals that all players except for $\plr_1$ behave in the following
way: the message sent by $\plr_i$ depends only on layers $1$ through
$i-1$ and the composition of layers $i+1$ through $k$. We say that
$\plr_i$ is {\em collapsing}. This notion is akin to that of the
aforementioned conservative protocols considered by Damm et al. Whereas
a conservative player composes the layers behind hers, a collapsing
player does so for layers ahead of hers. 

We consider what happens if we require {\em all} players in the protocol
to be collapsing. We prove a strong linear lower bound, showing that
even a single non-collapsing player makes an asymptotic difference in
the communication complexity.

\begin{thm} \label{thm:collapsing-lb}
  In a protocol for $\mpj_k$ where every player is collapsing, some
  player must communicate at least $n - \frac12\log n - 2 = n - O(\log n)$ 
  bits.
\end{thm}

Finally, one might wonder whether the collapsing requirement is so
strong that nothing nontrivial is possible anyway. The same question can
be raised for the conservative and myopic models where $\Omega(n/k^2)$
and $\Omega(n/k)$ lower bounds were proven in past work. It turns out
that the upper bound on $C(\mpjh_k)$ due to Damm et al.~\cite{DammJS98}
(see Section~\ref{sec:prev}) is achievable by a protocol that is both
conservative and myopic. We can show a similar upper bound via a
different protocol where every player is collapsing.

\begin{thm} \label{thm:collapsing-ub}
  For $k \ge 3$, there is an $O(n \log^{(k-1)}n)$-communication protocol
  for $\mpjhperm_k$ in which every player is collapsing. Here
  $\mpjhperm_k$ denotes the subproblem of $\mpjh_k$ in which layers $2$
  through $k$ of the input graph are permutations of $[n]$.
\end{thm}

The requirement that layers be permutations is a natural one and is not
new. The protocol of Pudlak et al.~also had this requirement; i.e., it
gave an upper bound on $C(\mpjperm_3)$. Theorem~\ref{thm:collapsing-ub}
can in fact be strengthened slightly by allowing one of the layers from
$2$ through $k$ to be arbitrary; we formulate and prove this stronger
version in Section~\ref{sec:collapsing-ub}.

\subsection{Organization}

The rest of the paper is organized as follows.
Theorems~\ref{thm:mpj-ub}, \ref{thm:collapsing-lb}
and~\ref{thm:collapsing-ub} are proven in Sections~\ref{sec:mpj-ub},
\ref{sec:collapsing-lb} and~\ref{sec:collapsing-ub} respectively.
Section~\ref{sec:prelim} introduces some notation that is used in
subsequent sections.


\section{A Sublinear Upper Bound} \label{sec:mpj-ub}

\subsection{Preliminaries, Notation and Overall Plan} \label{sec:prelim}

For the rest of the paper, ``protocols'' will be assumed to be
deterministic one-way NOF protocols unless otherwise qualified. We shall
use $\cost(P)$ to denote the total number of bits communicated in $P$,
for a worst case input.

Let us formally define the problems $\mpj_k$ and $\mpjh_k$. We shall
typically write the input $k$-tuple for $\mpj_k$ as
$(i,f_2,\ldots,f_{k-1},x)$ and that for $\mpjh_k$ as
$(i,f_2,\ldots,f_k)$, where $i\in[n]$, each $f_j\in\nn$ and $x\in\b^n$.
We then define $\mpj_k: [n]\times\left(\nn\right)^{k-2}\times \b^n\to\b$
and $\mpjh_k: [n]\times\left(\nn\right)^{k-1}\to[n]$ as follows.
\begin{alignat*}{2}
  \mpj_2(i, x) &:= x_i \, ; & \quad
  \mpj_k(i, f_2, f_3, \ldots, f_{k-1}, x) & := 
  \mpj_{k-1}(f_2(i), f_3, \ldots, f_{k-1}, x)\, , \mbox{~for~} k \ge 3\, \\
  \mpjh_2(i, f) &:= f(i) \, ; & \quad
  \mpjh_k(i, f_2, f_3, \ldots, f_k) & := 
  \mpjh_{k-1}(f_2(i), f_3, \ldots, f_k) \, , \mbox{~for~} k \ge 3\, .
\end{alignat*}
Here, $x_i$ denotes the $i$th bit of the string $x$. It will be helpful,
at times, to view strings in $\b^n$ as functions from $[n]$ to $\b$ and
use functional notation accordingly.  It is often useful to discuss the
composition of certain subsets of the inputs.  Let $\hat{i}_2 := i$, and
for $3 \le j \le k$, let $\hat{i}_j := f_{j-1}\circ \cdots \circ
f_2(i)$.  Similarly, let $\hat{x}_{k-1} := x$, and for $1 \le j \le
k-2$, let $\hat{x}_j := x\circ f_{k-1} \circ \cdots \circ f_{j+1}$.
Unrolling the recursion in the definitions, we see that, for $k
\ge 2$,
\begin{gather} 
  \mpj_k(i,f_2,\ldots,f_{k-1},x) 
  ~=~ x\circ f_{k-1}\circ\cdots\circ f_2(i)
  ~=~ \hat{x}_1(i) 
  ~=~ x_{\hat{i}_k} \, ; \label{eq:mpj-def} \\
  \mpjh_k(i,f_2,\ldots,f_k) 
  ~=~ f_k\circ\cdots\circ f_2(i) 
  ~=~ f_k(\hat{i}_k) \, . \label{eq:mpjh-def}
\end{gather}

We also consider the subproblems $\mpjperm_k$ and $\mpjhperm_k$ where
each $f_j$ above is a bijection from $[n]$ to $[n]$ (equivalently, a
permutation of $[n]$).  We let $\cS_n$ denote the set of all
permutations of $[n]$. 

Here is a rough plan of the proof of our sublinear upper bound. We
leverage the fact that a protocol $P$ for $\mpjperm_3$ with sublinear
communication is known. To be precise:
\begin{fact}[{Pudlak, R\"odl and Sgall~\cite[Corollary~4.8]{PudlakRS97}}]
\label{fact:prs-ub}
  $C(\mpjperm_3) = O(n\log\log n/\log n)$.
\end{fact}

The exact structure of $P$ will not matter; we shall only use $P$ as a
black box. To get a sense for why $P$ might be useful for, say,
$\mpj_3$, note that the players could replace $f_2$ with a permutation
$\pi$ and just simulate $P$, and this would work if $\pi(i) =
f(i)$. Of course, there is no way for $\plr_1$ and $\plr_3$ to agree
on a suitable $\pi$ without communication. However, as we shall see
below, it is possible for them to agree on a small enough {\em set} of
permutations such that either some permutation in the set is suitable,
or else only a small amount of side information conveys the desired
output bit to $\plr_3$.

This idea eventually gives us a sublinear protocol for $\mpj_3$.
Clearly, whatever upper bound we obtain for $\mpj_3$ applies to
$\mpj_k$ for all $k \ge 3$. However, we can decrease the upper bound
as $k$ increases, by embedding several instances of $\mpj_3$ into
$\mpj_k$. For clarity, we first give a complete proof of
Theorem~\ref{thm:mpj-ub} for the case $k = 3$.

\subsection{A $\mathbf{3}$-Player Protocol}

Following the plan outlined above, we prove Theorem~\ref{thm:mpj-ub} for
the case $k = 3$ by plugging Fact~\ref{fact:prs-ub} into the following
lemma, whose proof is the topic of this section.

\begin{lem} \label{lem:mpj-three-ub}
  Suppose $\phi:\Z^+\to(0,1]$ is a function such that $C(\mpjperm_3) =
  O(n\phi(n))$. Then $C(\mpj_3) = O(n\sqrt{\phi(n)})$.
\end{lem}

\begin{defn} \label{def:perm-cover}
  A set $\cA\ceq\cS_n$ of permutations is said to $d$-cover a function
  $f:[n]\to[n]$ if, for each $r\in[n]$, at least one of the following
  conditions holds:
  \begin{itemize}
    \item[(i)] $\exists\, \pi\in\cA$ such that $\pi(r) = f(r)$, or
    \item[(ii)] $|f^{-1}(f(r))| > d$.
  \end{itemize}
\end{defn}

\begin{lem} \label{lem:perm-cover}
  Let $f:[n]\to[n]$ be a function and $d$ be a positive integer. There
  exists a set $\cA_d(f)\ceq\cS_n$, with $|\cA_d(f)| \le d$, that
  $d$-covers $f$.
\end{lem}
\begin{proof}
  We give an explicit algorithm to construct $\cA_d(f)$. Our strategy is
  to partition the domain and codomain of $f$ (both of which equal
  $[n]$) into parts of matching sizes and then define bijections between
  the corresponding parts. To be precise, suppose $\Range(f) = \{s_1,
  s_2, \ldots, s_t\}$. Let $A_i = f^{-1}(s_i)$ be the corresponding
  fibers of $f$. Clearly, $\{A_i\}_{i=1}^t$ is a partition of $[n]$. It
  is also clear that there exists a partition $\{B_i\}_{i=1}^t$ of $[n]$
  such that, for all $i\in[t]$, $B_i \cap \Range(f) = \{s_i\}$ and
  $|B_i| = |A_i|$. We shall now define certain bijections $\pi_{i,\ell}:
  A_i \to B_i$, for each $i\in[t]$ and $\ell\in[d]$.

  Let $a_{i,1} < a_{i,2} < \cdots < a_{i,|A_i|}$ be the elements of
  $A_i$ arranged in ascending order. Similarly, let $b_{i,1} < \cdots <
  b_{i,|B_i|}$ be those of $B_i$. We define
  \[
    \pi_{i,\ell}(a_{i,j}) ~:=~ b_{i,(j-\ell)\bmod|B_i|} \, , \quad
    \mbox{for~} i\in[t], \ell\in[d] \, ,
  \]
  where, for convenience, we require ``$\alpha \bmod \beta$'' to return
  values in $[\beta]$, rather than $\{0,1,\ldots,\beta-1\}$.  It is
  routine to verify that $\pi_{i,\ell}$ is a bijection.  Notice that
  this construction ensures that for all $i\in[t]$ and $j\in[|A_i|]$ we
  have
  \begin{equation} \label{eq:perm-cover}
    |\{\pi_{i,\ell}(a_{i,j}):\, \ell\in[d]\}| ~=~ \min\{d, |B_i|\} \, .
  \end{equation}
  Let $\pi_\ell:[n]\to[n]$ be the bijection given by taking the
  ``disjoint union'' of $\pi_{1,\ell}, \ldots, \pi_{t,\ell}$. We claim
  that $\cA_d(f) = \{\pi_1, \ldots, \pi_d\}$ satisfies the conditions of
  the lemma.

  It suffices to verify that this choice of $\cA_d(f)$ $d$-covers $f$,
  i.e., to verify that every $r\in[n]$ satisfies at least one of the two
  conditions in Definition~\ref{def:perm-cover}.  Pick any $r\in[n]$.
  Suppose $r\in A_i$, so that $f(r)\in B_i$ and $\pi_\ell(r) =
  \pi_{i,\ell}(r)$. If $|B_i| > d$, then $|f^{-1}(f(r))| = |A_i| = |B_i|
  > d$, so condition~(ii) holds.  Otherwise, from Eq.~\eqref{eq:perm-cover},
  we conclude that $\{\pi_{i,\ell}(r):\, \ell\in[d]\} = B_i$. Therefore,
  for each $s\in B_i$ --- in particular, for $s = f(r)$ --- there exists
  an $\ell\in[d]$ such that $\pi_\ell(r) = \pi_{i,\ell}(r) = s$, so
  condition~(i) holds.
\end{proof}

\begin{proof}[Proof of Lemma~\ref{lem:mpj-three-ub}]

  Let $(i,\pi,x) \in [n]\times\cS_n\times\b^n$ denote an input for the
  problem $\mpjperm_3$. Then the desired output is $x_{\pi(i)}$. The
  existence of a protocol $P$ for $\mpjperm_3$ with $\cost(P) =
  O(n\phi(n))$ means that there exist functions
  \begin{gather*}
    \alpha:\cS_n\times\b^n\to\b^m \, , ~~
    \beta:[n]\times\b^n\times\b^m\to\b^m \, , \mbox{~~and~} \\
    \gamma:[n]\times\cS_n\times\b^m\times\b^m\to\b \, ,
  \end{gather*}
  where $m = O(n\phi(n))$, such that
  $\gamma(i,\pi,\alpha(\pi,x),\beta(i,x,\alpha(\pi,x))) = x_{\pi(i)}$.
  The functions $\alpha, \beta$ and $\gamma$ yield the messages in $P$
  of $\plr_1, \plr_2$ and $\plr_3$ respectively. 

  To design a protocol for $\mpj_3$, we first let $\plr_1$ and $\plr_3$
  agree on a parameter $d$, to be fixed below, and a choice of
  $\cA_d(f)$ for each $f:[n]\to[n]$, as guaranteed by
  Lemma~\ref{lem:perm-cover}. Now, let $(i,f,x) \in
  [n]\times\nn\times\b^n$ be an input for $\mpj_3$. Our protocol works
  as follows.
  \begin{itemize}
    \item $\plr_1$ sends a two-part message. The first part consists of
    the strings $\{\alpha(\pi,x)\}_\pi$ for all $\pi\in\cA_d(f)$. The
    second part consists of the bits $x_s$ for $s\in[n]$ such that
    $|f^{-1}(s)| > d$.
    \item $\plr_2$ sends the strings $\{\beta(i,x,\alpha)\}_\alpha$ for 
    all strings $\alpha$ in the first part of $\plr_1$'s message.
    \item $\plr_3$ can now output $x_{f(i)}$ as follows. If
    $|f^{-1}(f(i))| > d$, then she reads $x_{f(i)}$ off from the second
    part of $\plr_1$'s message. Otherwise, since $\cA_d(f)$ $d$-covers
    $f$, there exists a $\pi_0\in\cA_d(f)$ such that $f(i) = \pi_0(i)$.
    She uses the string $\alpha_0 := \alpha(\pi_0,x)$ from the first
    part of $\plr_1$'s message and the string $\beta_0 :=
    \beta(i,x,\alpha_0)$ from $\plr_2$'s message to output
    $\gamma(i,\pi_0,\alpha_0,\beta_0)$.
  \end{itemize}

  To verify correctness, we only need to check that $\plr_3$'s output in
  the ``otherwise'' case indeed equals $x_{f(i)}$. By the correctness of
  $P$, the output equals $x_{\pi_0(i)}$ and we are done, since $f(i) =
  \pi_0(i)$.

  We now turn to the communication cost of the protocol. By the
  guarantees in Lemma~\ref{lem:perm-cover}, $|\cA_d(f)| \le d$, so the
  first part of $\plr_1$'s message is at most $dm$ bits long, as is
  $\plr_2$'s message. Since there can be at most $n/d$ values $s\in[n]$
  such that $|f^{-1}(s)| > d$, the second part of $\plr_2$'s message is
  at most $n/d$ bits long. Therefore the communication cost is at most
  $2dm + n/d = O(dn\phi(n) + n/d)$. Setting $d =
  \ceil{1/\sqrt{\phi(n)}}$ gives us a bound of $O(n\sqrt{\phi(n)})$, as
  desired.
\end{proof}
\subsection{A $\mathbf{k}$-Player Protocol}

We now show how to prove Theorem~\ref{thm:mpj-ub} by generalizing the
protocol from Lemma~\ref{lem:mpj-three-ub} into a protocol for $k$
players. It will help to view an instance of $\mpj_k$ as incorporating
several ``embedded'' instances of $\mpj_3$. The following lemma makes
this precise.
\begin{lem} \label{lem:embed-three-into-k}
  Let $(i,f_2,\ldots,f_{k-1},x)$ be input for $\mpj_k$.  Then, for all
  $1<j<k$, \[\mpj_k(i,f_2,\ldots,x) = \mpj_3(f_{j-1}\circ\cdots\circ
  f_2(i),f_j,x\circ f_{k-1}\circ\cdots\circ f_{j+1}).\]
\end{lem}

In our protocol for $\mpj_k$, for $2\le j\le k-1$, the players
$\plr_1,\plr_j,$ and $\plr_k$ will use a modified version of the
protocol from Lemma~\ref{lem:mpj-three-ub} for $\mpj_3$ on input
$(f_{j-1}\circ\cdots\circ f_2(i), f_j, x\circ\cdots\circ f_{j+1})$.
Before we get to the protocol, we need to generalize the technical
definition and lemma from the previous subsection.

\begin{defn} \label{def:perm-s-cover}
  Let $S\ceq[n]$ and let $d$ be a positive integer.
  A set $\cA\ceq\cS_n$ of permutations is said to $(S,d)$-cover a function
  $f:[n]\to[n]$ if, for each $r\in S$, at least one of the following
  conditions holds:
  \begin{itemize}
    \item[(i)] $\exists\, \pi\in\cA$ such that $\pi(r) = f(r)$, or
    \item[(ii)] $|S \cap f^{-1}(f(r))| > d$.
  \end{itemize}
\end{defn}

\begin{lem} \label{lem:perm-s-cover}
  Let $f:[n]\to[n]$ be a function, $S \subseteq [n]$, and $d$ be a
  positive integer. There exists a set $\cA_{S,d}(f)\ceq\cS_n$, with
  $|\cA_{S,d}(f)| \le d$, that $(S,d)$-covers $f$.
\end{lem}
\begin{proof}
  This proof closely follows that of Lemma~\ref{lem:perm-cover}.
  As before, we give an explicit algorithm to construct
  $\cA_{S,d}(f)$. Suppose $\Range(f) = \{s_1, s_2, \ldots, s_t\}$, and
  let $\{A_i\}$ and $\{B_i\}$ be defined as in
  Lemma~\ref{lem:perm-cover}.  Let $a_{i,1} < \cdots < a_{i,z}$ be the
  elements of $A_i \cap S$ arranged in ascending order, and let
  $a_{i,z+1} < \cdots < a_{i,|A_i|}$ be the elements of $A_i \setminus
  S$ arranged in ascending order.  Similarly, let $b_{i,1} < \cdots <
  b_{i,|B_i|-1}$ be the elements of $B_i\setminus \{s_i\}$ arranged in
  ascending order, and let $b_{i,|B_i|} = s_i$.  For $i \in [t],
  \ell\in[d]$, we define $\pi_{i,\ell}(a_{i,j}) ~:=~
  b_{i,(j-\ell)\bmod|B_i|}$.
  As before, it is routine to verify that $\pi_{i,\ell}$ is a bijection.
  Let $\pi_\ell:[n]\to[n]$ be the bijection given by taking the
  ``disjoint union'' of $\pi_{1,\ell}, \ldots, \pi_{t,\ell}$. We claim
  that $\cA_{S,d}(f) = \{\pi_1, \ldots, \pi_d\}$ satisfies the conditions of
  the lemma.

  It suffices to verify that this choice of $\cA_{S,d}(f)$
  $(S,d)$-covers $f$, i.e., to verify that every $r\in S$ satisfies at
  least one of the two conditions in
  Definition~\ref{def:perm-s-cover}.  Pick any $r\in S$.  Suppose
  $r\in A_i$, and fix $j$ such that $r = a_{i,j}$.  If $|S\cap
  f^{-1}(f(r))| > d$, then condition~(ii) holds.  Otherwise, setting
  $\ell = j < |S\cap f^{-1}(f(i))| \le d$, we conclude that
  $\pi_\ell(r) = \pi_{i,\ell}(r) = \pi_{i,\ell}(a_{i,j}) =
  b_{i,|B_i|} = s_i = f(r)$, so condition~(i) holds.
\end{proof}

\begin{proof}[Proof of Theorem~\ref{thm:mpj-ub}]
  To design a protocol for $\mpj_k$, we first let $\plr_1$ and
  $\plr_k$ agree on a parameter $d$, to be fixed below.  They also
  agree on a choice of $\mathcal{A}_{S,d}(f)$ for all $S \subseteq
  [n]$ and $f : [n] \rightarrow [n]$.
  Let $(i,f_2,\ldots,f_{k-1},x)$ denote an input for $\mpj_k$.  Also,
  let $S_1 = [n]$, and for all $2 \le j \le k-1$, let $S_j = \{s \in
  [n]: |S_{j-1}\cap f_j^{-1}(s)| > d\}$.  Our protocol works as
  follows:
  \begin{itemize}
    \item $\plr_1$ sends a $(k-1)$-part message.  For $1 \le j \le k-2$,
      the $j$th part of $\plr_1$'s message consists of the strings
      $\{\alpha(\pi,\hat{x}_{j+1})\}_\pi$ for each $\pi \in
      \mathcal{A}_{S_j,d} (f_{j+1})$.  
      The remaining part consists
      of the bits $x_s$ for $s \in S_{k-1}$.  
    \item For $2\le j \le k-1$, $\plr_j$ sends the strings
      $\{\beta(\hat{i}_j, \hat{x}_j, \alpha)\}_\alpha$ for all strings
      $\alpha$ in the $(j-1)$th part of $\plr_1$'s message.  
    \item $\plr_k$ can now output $x_{\hat{i}_k}$ as follows.  If
    $|S_1 \cap f_2^{-1}(f_2(i))| \le d$, then, because 
    $\mathcal{A}_{S_1,d}(f_2)\ (S_1,d)$-covers $f_2$, there exists
    $\pi_0 \in \mathcal{A}_{S_1,d}(f_2)$ such that $f_2(i) =
    \pi_0(i)$.  She uses the string $\alpha_0 = \alpha(\pi_0,
    \hat{x}_2)$ from the first part of $\plr_1$'s message and the
    string $\beta_0 = \beta(i,\hat{x}_2, \alpha_0)$ from $\plr_2$'s
    message to output $\gamma_0 = \gamma(i,\pi_0,\alpha_0, \beta_0)$.
    Similarly, if there is a $j$ such that $2\le j\le k-2$ and $|S_j \cap
    f_{j+1}^{-1}(f_{j+1}(\hat{i}_{j+1}))| \le d$, then since
    $\mathcal{A}_{S_j,d}(f_{j+1})\ (S_j,d)$-covers $f_{j+1}$,
    there exists a $\pi_0 \in \mathcal{A}_{S_j,d}(f_{j+1})$ such that
    $f_{j+1}(\hat{i}_{j+1}) = \pi_0(\hat{i}_{j+1})$.  She uses the
    string $\alpha_0 = \alpha(\pi_0, \hat{x}_{j+1})$ from the $j$th
    part of $\plr_1$'s message and the string $\beta_0 =
    \beta(\hat{i}_{j+1},\hat{x}_{j+1}, \alpha_0)$ from $\plr_{j+1}$'s
    message to output $\gamma_0 = \gamma(\hat{i}_{j+1}, \pi_0,
    \alpha_0, \beta_0)$.  Otherwise, $|S_{k-2} \cap
    f_{k-1}^{-1}(f_{k-1}(\hat{i}_{k-1}))| > d$, hence $\hat{i}_{k} \in
    S_{k-1}$, and she reads $x_{\hat{i}_k}$ off from the last part of
    $\plr_1$'s message.
  \end{itemize}

  To verify correctness, we need to ensure that $\plr_k$ always
  outputs $x\circ f_{k-1} \circ \cdots \circ f_2(i)$.  In the following
  argument, we repeatedly use Lemma~\ref{lem:embed-three-into-k}.  We 
  proceed inductively.  If $|S_1 \cap f_2^{-1}(f_2(i))| \le d$ then there
  exists $\pi_0 \in \mathcal{A}_{S_1,d}(f_2)$ such that $f_2(i) =
  \pi_0(i)$, $\alpha_0 = \alpha(\pi_0, \hat{x}_2)$, and $\beta_0 =
  \beta(i, \hat{x}_2, \alpha_0)$, and $\plr_k$ outputs $\gamma_0 =
  \gamma(i, \pi_0, \alpha_0, \beta_0) = \hat{x}_2(\pi_0(i)) = x \circ
  f_{k-1} \circ \cdots \circ f_2(i)$.  Otherwise, $|S_1 \cap
  f_2^{-1}(f_2(i))| > d$, hence $f_2(i) \in S_2$.  Inductively, if
  $\hat{i}_j \in S_{j-1}$, then either $|S_{j-1}\cap
  f_j^{-1}(f_j(\hat{i}_j))| \le d$, or $|S_{j-1}\cap
  f_j^{-1}(f_j(\hat{i}_j))| > d$.  In the former case, there is $\pi_0
  \in \mathcal{A}_{S_{j-1},d}(f_j)$ such that $f_j(\hat{i}_j) =
  \pi_0(\hat{i}_j)$; $\alpha_0(\pi_0, \hat{x}_j)$, and $\beta_0 =
  \beta(\hat{i}_j, \hat{x}_j, \alpha_0)$, and $\plr_k$ outputs
  $\gamma_0 = \gamma(\hat{i}_j, \pi_0, \alpha_0, \beta_0) =
  \hat{x}_j(f_j(\hat{i}_j)) = x \circ f_{k-1} \circ \cdots \circ
  f_2(i)$.  In the latter case, $f_j(\hat{i}_j) \in S_j$.  By
  induction, we have that either $\plr_k$ outputs $x\circ f_{k-1}
  \circ \cdots \circ f_2(i)$, or $\hat{i}_k \in S_{k-1}$.  But in this
  case, $\plr_k$ outputs $x(\hat{i}_k) = x \circ f_{k-1} \circ \cdots \circ f_2(i)$
  directly from the last part of $\plr_1$'s message.  Therefore,
  $\plr_k$ always outputs $x\circ f_{k-1} \circ \cdots \circ f_2(i)$
  correctly.

  We now turn to the communication cost of the protocol.  By
  Lemma~\ref{lem:perm-s-cover}, $|\mathcal{A}_{S_j, d}(f_j)| \le d$
  for each $2 \le j \le k-1$, hence the first $k-2$ parts of
  $\plr_1$'s message each are at most $dm$ bits long, as is $\plr_j$'s
  message for all $2 \le j \le k-1$.  Also, since for all $2 \le j \le
  k-1$, there are at most $|S_{j-1}|/d$ elements $s \in S_j$ such that
  $|S_{j-1}\cap f_j^{-1}(s)| > d$, we must have that $|S_2| \le
  |S_1|/d = n/d, |S_3| \le |S_2|/d \leq n/d^2, $ etc., and $|S_{k-1}| \le
  n/d^{k-2}$.  Therefore, the final part of $\plr_1$'s message is at
  most $n/d^{k-2}$ bits long, and the total communication cost is at
  most $2(k-2)dm + n/d^{k-2} = O((k-2)dn\phi(n) + n/d^{k-2})$.
  Setting $d = \lceil 1/((k-2)\phi(n))^{1/(k-1)} \rceil$ gives us a bound of
  $O(n(k\phi(n))^{(k-2)/(k-1)})$ as desired.
\end{proof}

Note that, in the above protocol, except for the first and last players,
the remaining players access very limited information about their input.
Specifically, for all $2 \leq j \leq k-1$, $\plr_j$ needs to see only
$\hat{i}_j$ and $\hat{x}_j$, i.e., $\plr_j$ is both \emph{conservative}
and \emph{collapsing}.  Despite this severe restriction, we have a
sublinear protocol for $\mpj_k$.  As we shall see in the next section,
further restricting the input such that $\plr_1$ is also collapsing
yields very strong lower bounds.


\section{Collapsing Protocols: A Lower Bound} \label{sec:collapsing-lb}

Let $F:\scrA_1\times\scrA_2\times\cdots\times\scrA_k\to\scrB$ be a
$k$-player NOF communication problem and $P$ be a protocol for $F$. We
say that $\plr_j$ is {\em collapsing} in $P$ if her message depends only
on $x_1,\ldots,x_{j-1}$ and the function
$g_{x,j}:\scrA_1\times\scrA_2\times\cdots\times\scrA_j\rightarrow\scrB$
given by $g_{x,j}(z_1,\ldots,z_j) =
F(z_1,\ldots,z_j,x_{j+1},\ldots,x_k)$.  For pointer jumping, this
amounts to saying that $\plr_j$ sees all layers $1,\ldots,{j-1}$ of
edges (i.e., the layers preceding the one on her forehead), but not
layers $j+1,\ldots,k$; however, she does see the result of following the
pointers from each vertex in layer $j$. Still more precisely, if the
input to $\mpj_k$ (or $\mpjh_k$) is $(i,f_2,\ldots,f_k)$, then the only
information $\plr_j$ gets is $i, f_2, \ldots, f_{j-1}$ and the
composition $f_k \circ f_{k-1} \circ \cdots \circ f_{j+1}$. 

We say that a protocol is collapsing if every player involved is
collapsing. We shall prove Theorem~\ref{thm:collapsing-lb} by
contradiction. Assume that there is a collapsing protocol $P$ for
$\mpj_k$ in which every player sends less than $n - \frac12\log n - 2$
bits. We shall construct a pair of inputs that differ only in the last
layer (i.e., the Boolean string on $\plr_k$'s forehead) and that cause
players $1$ through $k-1$ to send the exact same sequence of messages.
This will cause $\plr_k$ to give the same output for both these inputs.
But our construction will ensure that the desired outputs are unequal, a
contradiction.  To aid our construction, we need some definitions and
preliminary lemmas.

\begin{defn} \label{def:consistent}
  A string $x\in\b^n$ is said to be {\em consistent} with
  $(f_1,\ldots,f_j,\alpha_1,\ldots,\alpha_j)$ if, in protocol $P$, for
  all $h \leq j$, $\plr_h$ sends the message $\alpha_h$ on seeing input
  $(i=f_1,\ldots,f_{h-1}, x\circ f_j \circ f_{j-1} \circ \cdots \circ
  f_{h+1})$ and previous messages
  $\alpha_1,\ldots,\alpha_{h-1}$.\footnote{It is worth noting that, in
  Definition~\ref{def:consistent}, $x$ is not to be thought of as an
  input on $\plr_k$'s forehead. Instead, in general, it is the
  composition of the rightmost $k-j$ layers of the input graph.} A
  subset $T\ceq\b^n$ is said to be consistent with
  $(f_1,\ldots,f_j,\alpha_1,\ldots,\alpha_j)$ if $x$ is consistent with
  $(f_1,\ldots,f_j,\alpha_1,\ldots,\alpha_j)$ for all $x \in T$.
\end{defn}

\begin{defn}
  For strings $x,x'\in\b^n$ and $a,b\in\b$, define the sets
  \[
    I_{ab}(x,x') ~:=~ \{j\in[n]:\, (x_j,x'_j) = (a,b)\} \, .
  \]
  A pair of strings $(x,x')$ is said to be a {\em crossing pair} if for
  all $a,b\in\b$, $I_{ab}(x,x') \ne \emptyset$.  A set $T\ceq\b^n$ is
  said to be {\em crossed} if it contains a crossing pair and {\em
  uncrossed} otherwise. The {\em weight} of a string $x\in\b^n$ is
  defined to be the number of $1$s in $x$, and denoted $|x|$.
\end{defn}

For the rest of this section, we assume (without loss of generality)
that $n$ is large enough and even.

\begin{lem} \label{lem:uncrossed-half-weight}
  If $T\ceq\b^n$ is uncrossed, then $|\{x\in T:\, |x| = n/2\}| \le 2$.
\end{lem}
\begin{proof}
  Let $x$ and $x'$ be distinct elements of $T$ with $|x| = |x'| = n/2$.
  For $a,b\in\b$, define $t_{ab} = |I_{ab}(x,x')|$.  Since $x \ne x'$,
  we must have $t_{01} + t_{10} > 0$. An easy counting argument shows
  that $t_{01} = t_{10}$ and $t_{00} = t_{11}$. Since $T$ is uncrossed,
  $(x,x')$ is not a crossing pair, so at least one of the numbers
  $t_{ab}$ must be zero. It follows that $t_{00} = t_{11} = 0$, so $x$
  and $x'$ are bitwise complements of each other. 
  Since this holds for any two strings in $\{x\in T:\, |x| = n/2\}$,
  that set can have size at most $2$.
\end{proof}

\begin{lem} \label{lem:crossed-part}
  Suppose $t \le n - \frac12\log n - 2$. If $\b^n$ is partitioned into
  $2^t$ disjoint sets, then one of those sets must be crossed.
\end{lem}
\begin{proof}
  Let $\b^n = T_1 \sqcup T_2 \sqcup \cdots \sqcup T_m$ be a partition of
  $\b^n$ into $m$ uncrossed sets. Define $X := \{x\in\b^n:\, |x| =
  n/2\}$.  Then $X = \bigcup_{i=1}^m (T_i \cap X)$. By
  Lemma~\ref{lem:uncrossed-half-weight},
  \[
    |X| ~\le~ \sum_{i=1}^m |T_i \cap X| ~\le~ 2m \, .
  \]
  Using Stirling's approximation, we can bound $|X| > 2^n/(2\sqrt n)$.
  Therefore, $m > 2^{n - \frac12\log n - 2}$.
\end{proof}

\begin{proof}[Proof of Theorem~\ref{thm:collapsing-lb}]
  Set $t = n - \frac12\log n - 2$. Recall that we have assumed that
  there is a collapsing protocol $P$ for $\mpj_k$ in which every player
  sends at most $t$ bits. We shall prove the following statement by
  induction on $j$, for $j\in[k-1]$.
  \begin{quote} \begin{itemize}
    \item[(*)] There exists a partial input
    $(i=f_1,f_2,\ldots,f_j)\in[n]\times\left(\nn\right)^{j-1}$, a
    sequence of messages $(\alpha_1, \ldots, \alpha_j)$ and a crossing
    pair of strings $(x,x')\in\left(\b^n\right)^2$ such that both $x$
    and $x'$ are consistent with $(f_1,\ldots,f_j,
    \alpha_1,\ldots,\alpha_j)$, whereas $x\circ f_j\circ\cdots\circ
    f_2(i) = 0$ and $x'\circ f_j\circ\cdots\circ f_2(i) = 1$.  
  \end{itemize} \end{quote}
  Considering~(*) for $j = k-1$, we see that $\plr_k$ must behave
  identically on 
  the two 
  inputs $(i,f_2,\ldots,f_{k-1},x)$ and
  $(i,f_2,\ldots,f_{k-1},x')$. Therefore, she must err on one of these
  two inputs. This will give us the desired contradiction.

  To prove~(*) for $j=1$, note that $\plr_1$'s message, being at most
  $t$ bits long, partitions $\b^n$ into at most $2^t$ disjoint sets. By
  Lemma~\ref{lem:crossed-part}, one of these sets, say $T$, must be
  crossed. Let $(x,x')$ be a crossing pair in $T$ and let $\alpha_1$ be
  the message that $\plr_1$ sends on seeing a string in $T$. Fix $i =
  f_1$ such that $i\in I_{01}(x,x')$. These choices are easily seen to
  satisfy the conditions in~(*).
  Now, suppose~(*) holds for a particular $j\ge 1$. Fix the partial
  input $(f_1,\ldots,f_j)$ and the message sequence
  $(\alpha_1,\ldots,\alpha_j)$ as given by~(*). We shall come up with
  appropriate choices for $f_{j+1}$, $\alpha_{j+1}$ and a new crossing
  pair $(y,y')$ to replace $(x,x')$, so that~(*) is satisfied for $j+1$.
  Since $\plr_{j+1}$ sends at most $t$ bits, she partitions $\b^n$ into
  at most $2^t$ subsets (the partition might depend on the choice of
  $(f_1,\ldots,f_j,\alpha_1,\ldots,\alpha_j)$).

  As above, by Lemma~\ref{lem:crossed-part}, she sends a message
  $\alpha_{j+1}$ on some crossing pair $(y,y')$. Choose $f_{j+1}$ so
  that it maps $I_{ab}(x,x')$ to $I_{ab}(y,y')$ for all $a,b\in\b$; this
  is possible because $I_{ab}(y,y') \ne \emptyset$.  Then, for all
  $i\in[n]$, $x_i = y_{f_{j+1}(i)}$ and $x'_i = y'_{f_{j+1}(i)}$. Hence,
  $x = y\circ f_{j+1}$ and $x' = y'\circ f_{j+1}$. Applying the
  inductive hypothesis and the definition of consistency, it is
  straightforward to verify the conditions of~(*) with these choices for
  $f_{j+1}, \alpha_{j+1}, y$ and $y'$. This completes the proof.
\end{proof}


\vskip-0.3cm
\section{Collapsing Protocols: An Upper Bound} \label{sec:collapsing-ub}

We now turn to proving Theorem~\ref{thm:collapsing-ub} by constructing
an appropriate collapsing protocol for $\mpjhperm_k$. Our protocol uses
what we call {\em bucketing schemes}, which have the flavor of the
conservative protocol of Damm et al.~\cite{DammJS98}.  For any
function $f \in [n]^{[n]}$ and any $S \subseteq [n]$, let $\ind{S}$
denote the indicator function for $S$; that is, $\ind{S}(i) = 1
\Leftrightarrow i \in S$.  Also, let $f|_S$ denote the function $f$
restricted to $S$; this can be seen as a list of numbers $\{i_s\}$,
one for each $s\in S$.  Players will often need to send $\ind{S}$ and
$f|_S$ together in a single message.  This is because later players
might not know $S$, and will therefore be unable to interpret $f|_S$
without $\ind{S}$.  Let $\ang{m_1,\ldots, m_t}$ denote the
concatenation of messages $m_1,\ldots, m_t$.  

\begin{defn}
  A {\em bucketing scheme} on a set $X$ is an ordered partition $\cB =
  (B_1,\ldots, B_t)$ of $X$ into {\em buckets}. For $x\in X$, we write
  $\cB[x]$ to denote the unique integer $j$ such that $B_j \ni x$.
\end{defn}

We actually prove our upper bound for problems slightly more general
than $\mpjhperm_k$. To be precise, for an instance $(i,f_2,\ldots,f_k)$
of $\mpjh_k$, we allow any one of $f_2, \ldots, f_k$ to be an arbitrary
function in $\nn$. The rest of the $f_j$s are required to be
permutations, i.e., in $\cS_n$. 

\begin{thm}[Slight generalization of Theorem~\ref{thm:collapsing-ub}]
\label{thm:one-layer-perm}
  There is an $O(n\log^{(k-1)}n)$ collapsing protocol for instance
  $(i,f_2,\ldots,f_k)$ of $\mpjh_k$ when all but one of
  $f_2,\ldots,f_k$ are permutations. In particular, there is such a
  protocol for $\mpjhperm_k$.
\end{thm}
\begin{proof}
  We prove this for $\mpjhperm_k$ only.
  For $1 \le t \le \ceil{\log n}$, 
  define the bucketing scheme $\cB_t = (B_1,\ldots,
  B_{2^t})$ on $[n]$ by $B_j := \{r \in [n]: \lceil 2^t r/n \rceil = j\}$.
  Note that each $|B_j| \leq \lceil n/2^t \rceil$ and that a bucket can
  be described using $t$ bits.  For $1 \leq j \leq k$, let $b_j =
  \ceil{\log^{(k-j)} n}$.  In the protocol, most players will use two
  bucketing schemes, $\cB$ and $\cB'$. On input $(i,f_2,\ldots,f_k)$:
  \begin{itemize}
    \item $\plr_1$
      sees $\hat{f}_1$, computes $\cB' := \cB_{b_1}$, and sends 
      $\ang{\cB'[\hat{f}_1(1)], \ldots, 
	\cB'[ \hat{f}_1(n)]}$.
    \item $\plr_2$ sees $\hat{i}_2, \hat{f}_2$, and $\plr_1$'s
      message.  $\plr_2$ computes $\cB := \cB_{b_1}$ and $\cB' :=
      \cB_{b_2}$.  She recovers $b := \cB[\hat{f}_2(f_2(\hat{i}_2))]$
      and hence $B_b$.  Let $S_2 := \{s \in [n]: \hat{f}_2(s) \in
      B_b\}$.  Note that $f_2(\hat{i}_2) \in S_2$.  $\plr_2$ sends
      $\ang{\ind{S_2}, \{\cB^\prime[ \hat{f}_2(s)]:s \in S_2\}}$.

    \centerline{\vdots}

    \item $\plr_j$ sees $\hat{i}_j, \hat{f}_j$, and $\plr_{j-1}$'s
      message.  $\plr_j$ computes $\cB := \cB_{b_{j-1}}$ and $\cB' :=
      \cB_{b_j}$.  She recovers $b := \cB[\hat{f}_j(f_j(\hat{i}_j))]$
      and hence $B_b$.  Let $S_j := \{s \in [n]: \hat{f}_j(s) \in
      B_b\}$.  Note that the definitions guarantee that
      $f_j(\hat{i}_j) \in S_j$.  $\plr_j$ sends $\ang{\ind{S_j},
      \{\cB^\prime[ \hat{f}_j(s)]:s \in S_j\}}$.

    \centerline{\vdots}

    \item $\plr_k$ sees $\hat{i}_k$ and $\plr_{k-1}$'s message and
    outputs $f_k(\hat{i}_k)$.
  \end{itemize}

  We claim that this protocol costs $O(n \log^{(k-1)} n)$ and
  correctly outputs $\mpjh_k(i,f_2,\ldots,f_k)$.  For each $2 \leq j
  \leq k-1$, $\plr_j$ uses bucketing scheme $\cB_{b_{j-1}}$ to recover
  the bucket $B_b$ containing $\hat{f}_j(f_j(\hat{i}_j))$.  She then
  encodes each element in $B_b$ in the bucketing scheme $\cB_{b_j}$.
  Each bucket in $\cB_{b_j}$ has size at most $\lceil
  n/b_{j+1}\rceil$.  In particular, each bucket in scheme $\cB_{k-1}$
  has size at most $\lceil n/b_k \rceil = 1$, and the unique element in 
  the bucket (if present) is precisely $f_k(\hat{i}_k)$.
  Turning to the communication cost, $\plr_1$ sends $b_1 =
  \ceil{\log^{(k-1)} n}$ bits to identify the bucket for each $i \in
  [n]$, giving a total of $n\ceil{\log^{(k-1)}n}$ bits.  For $1 < j <
  k$, $\plr_j$ uses $n + b_j(n/b_j) = O(n)$ bits.  Thus, the total cost is
  $O(n\log^{(k-1)} n + kn)$ bits.

  For $k \le \log^* n$ players, we are done. For larger $k$, we can get
  an $O(n)$ protocol by doubling the size of each $b_j$ and stopping the
  protocol when the buckets have size $\leq 1$.
\end{proof}


\vskip-0.3cm
\section{Concluding Remarks}
We have presented the first nontrivial upper bound on the NOF
communication complexity of the Boolean problem $\mpj_k$, showing that
$C(\mpj_k) = o(n)$. A lower bound of $\Omega(n)$ had seemed {\em a
priori} reasonable, but we show that this is not the case. One
plausible line of attack on lower bounds for $\mpj_k$ is to treat it
as a {\em direct sum} problem: at each player's turn, it seems that
$n$ different paths need to be followed in the input graph, so it
seems that an information theoretic approach (as in Bar-Yossef et
al.~\cite{BarYossefJKS02} or Chakrabarti~\cite{Chakrabarti07}) could
lower bound $C(\mpj_k)$ by $n$ times the complexity of some simpler
problem. However, it appears that such an approach would naturally
yield a lower bound of the form $\Omega(n/\xi(k))$, as in
Conjecture~\ref{conj:mpj-lin-lb}, which we have explicitly falsified.

The most outstanding open problem regarding $\mpj_k$ is to resolve
Conjecture~\ref{conj:mpj-accz}. A less ambitious, but seemingly
difficult, goal is to get tight bounds on $C(\mpj_3)$, closing the gap
between our $O(n\sqrt{\log\log n/\log n})$ upper bound and Wigderson's
$\Omega(\sqrt n)$ lower bound.  A still less ambitious question is prove
that $\mpj_3$ is harder than its very special subproblem $\tpj_3$
(defined in Section~\ref{sec:prev}).  Our $n - O(\log n)$ lower bound
for collapsing protocols is a step in the direction of improving the
known lower bounds. We hope our technique provides some insight about
the more general problem.


  \bibliographystyle{alpha}


\end{document}